\documentclass[letterpaper, 10 pt, conference]{ieeeconf}  

\usepackage{multicol}
\usepackage[bookmarks=true]{hyperref}
\IEEEoverridecommandlockouts  
\usepackage{cite}
\usepackage{amsmath,amssymb,amsfonts, dsfont}
\usepackage{algorithmic}
\usepackage{graphicx}
\usepackage{textcomp}
\usepackage{setspace}
\usepackage{booktabs}

\usepackage{epsfig}
\usepackage{times} 
\usepackage{svg}
\usepackage[linesnumbered,ruled,vlined]{algorithm2e} 

\usepackage[nolist, nohyperlinks]{acronym}

\usepackage{bm}
\usepackage{breqn}
\usepackage{amsmath}
\usepackage{amssymb}
\usepackage{tabularx}

\usepackage{amsthm}

\newcommand{\reals}{\mathbb{R}}
\newcommand{\R}{\reals}
\newcommand{\Rnonneg}{\reals_{\geq 0}}

\newcommand{\naturals}{\mathbb{N}}
\newcommand{\N}{\naturals}

\newcommand{\Dcal}{\mathcal{D}}

\newcommand{\Kcal}{\mathcal{K}}

\newcommand{\Scal}{\mathcal{S}}

\newcommand{\Ucal}{\mathcal{U}}

\newcommand{\Xcal}{\mathcal{X}}

\newcommand{\eqn}[1]{\begin{align} #1 \end{align}}

\theoremstyle{plain}
\newtheorem{theorem}{Theorem}

\newtheorem{problem}{Problem}
\newtheorem{definition}{Definition}

\theoremstyle{definition}
\newtheorem{assumption}{Assumption}
\newtheorem{remark}{Remark}

\theoremstyle{remark}

\makeatletter
\let\NAT@parse\undefined
\makeatother
\usepackage{hyperref}
\hypersetup{
colorlinks, 
    linkcolor={red!50!black},
    citecolor={blue!50!black},
    urlcolor={blue!80!black}
}
\usepackage{cleveref}

\crefformat{problem}{Problem~#2#1#3}
\crefformat{assumption}{Assumption~#2#1#3}

\usepackage[nolist,nohyperlinks]{acronym}

\acrodef{LIP}[LIP]{linear-in-parameter}
\acrodef{OCP}[OCP]{online conformal prediction}
\acrodef{DNN}[DNN]{deep neural network}
\acrodef{CBF}[CBF]{control barrier function}
\acrodef{MPC}[MPC]{model predictive control}
\acrodef{SIOCP}[SI-OCP]{staggered integral online conformal prediction}

\usepackage{tikz}

\title{\LARGE \bf Staggered Integral Online Conformal Prediction for Safe Dynamics Adaptation with Multi-Step Coverage Guarantees 
}
\crefname{definition}{Def.}{Defs.}

\author{Daniel M. Cherenson and Dimitra Panagou 
\thanks{This research was supported by the Center for Autonomous Air Mobility and Sensing (CAAMS), an NSF IUCRC, under Award Number 2137195, and an NSF CAREER under Award Number 1942907.}
\thanks{All authors are with the Robotics Department, University of Michigan, Ann Arbor, MI, USA
        {\tt \{dmrc, dpanagou\}@umich.edu}}%
\thanks{Code: \href{https://github.com/dcherenson/staggered-integral-ocp}{https://github.com/dcherenson/staggered-integral-ocp}}
}

\begin{document}

\maketitle
\begin{abstract}
 Safety-critical control of uncertain, adaptive systems often relies on conservative, worst-case uncertainty bounds that limit closed-loop performance. Online conformal prediction is a powerful data-driven method for quantifying uncertainty when truth values of predicted outputs are revealed online; however, for systems that adapt the dynamics without measurements of the state derivatives, standard online conformal prediction is insufficient to quantify the model uncertainty. We propose Staggered Integral Online Conformal Prediction (SI-OCP), an algorithm utilizing an integral score function to quantify the lumped effect of disturbance and learning error. This approach provides long-run coverage guarantees, resulting in long-run safety when synthesized with safety-critical controllers, including robust tube model predictive control. Finally, we validate the proposed approach through a numerical simulation of an all-layer deep neural network (DNN) adaptive quadcopter using robust tube \ac{MPC}, highlighting the applicability of our method to complex learning parameterizations and control strategies. 

\end{abstract}

\section{Introduction}
Online dynamics adaptation is a promising route to reducing conservatism in safety-critical control because it allows the prediction model to track time-varying disturbances and modeling error during execution. This is particularly valuable for predictive controllers, whose closed-loop performance depends critically on how tightly the future model mismatch can be bounded. Existing safe adaptive control methods have made important progress, including set-membership and zonotopic uncertainty estimation for adaptive control barrier functions and robust adaptive MPC constructions that exploit structural assumptions on the uncertainty~\cite{lopez2020robust,fan2020deep, black2021fixed,cohen2024uncertainty, tao2024robust}. These approaches provide meaningful guarantees, but they typically rely on a priori disturbance bounds, restrictive parameterizations, or distributional assumptions, all of which can become overly conservative when the model class is expressive and the environment evolves online.

Conformal prediction offers an attractive complementary viewpoint. Classical conformal prediction is distribution-free and model-agnostic, converting observed prediction errors into calibrated uncertainty sets~\cite{lindemann2023safe,angelopoulos2023conformal}. For streaming, temporally correlated data, \ac{OCP}, also referred to as adaptive conformal inference, extends this idea beyond exchangeable settings and provides long-run coverage guarantees for non-i.i.d. sequences~\cite{gibbs2021adaptive, angelopoulos2024online, areces2025online}. Recent work has successfully used \ac{OCP} in robotics and safe control to quantify the motion uncertainty of other agents, bound discrete disturbance uncertainty, and synthesize safe learning-based controllers~\cite{dixit2023adaptive, zhou2024safety, zhang2025safety, zhou2025computationally}. These results highlight that \ac{OCP} is useful because it is lightweight, nonparametric, and compatible with complex learned predictors.

However, existing \ac{OCP} formulations are not directly suited to online dynamics adaptation for continuous-time systems. For safety-critical applications, one must quantify the residual dynamics error between the true system and the learned model, since this error determines the robustness margin required to guarantee safety. Standard \ac{OCP} would require a ground-truth label for this residual after each prediction, which is not directly available online. Evaluating the instantaneous residual requires the state derivative, whereas only state measurements are available. Thus, the main challenge is not only distribution shift, but the \textbf{absence of directly observed labels} for the adaptive dynamics residual itself.

A second challenge is that the safety objective is inherently multi-step because the controller reasons over a finite prediction horizon. Existing \ac{OCP}-based safe control applications work well in their intended settings, but they only \textbf{provide one-step or instantaneous coverage guarantees}~\cite{dixit2023adaptive, zhou2024safety, zhang2025safety, zhou2025computationally}. In contrast, predictive safety-critical controllers require a robustness margin that remains valid over a future horizon, since recursive feasibility and safety depend on the accumulated effect of model error along the planned trajectory. While we leverage the adaptive model to predict the dynamics in continuous time, our contribution is the \textbf{multi-step conformal coverage} framework for a \emph{general class of adaptive models}. Specifically, we assess the adaptive model prediction over the most recent horizon and use that delayed information to calibrate uncertainty over the next horizon.

In this paper, we address these two gaps with \textbf{Staggered Integral Online Conformal Prediction (SI-OCP)}, a distribution-free uncertainty quantification framework for safe online dynamics adaptation. Instead of relying on unavailable instantaneous labels, we propose a new integral non-conformity score computed from a rolling window of measured state, input, and parameter trajectories. This score is efficiently computed online and captures the lumped effect of approximation error, adaptation error, and exogenous disturbance. Because the label associated with a horizon-level prediction becomes available only after that horizon has elapsed, we maintain multiple conformal thresholds and update them in a staggered manner as delayed labels are revealed. 
At each time step, \acs{SIOCP} calibrates a robustness margin for the future horizon and provides asymptotic long-run multi-step coverage of the true disturbance.

When this conformalized margin is applied to a robust predictive safety-critical controller, the resulting closed-loop system achieves asymptotic long-run safety. The framework is agnostic to the underlying adaptation law and therefore applies broadly to online dynamics predictors, ranging from classical linear-in-parameter models~\cite{parikh2019integral, black2021fixed,cohen2024uncertainty,zhang2025safety} to recent all-layer \ac{DNN} adaptation schemes~\cite{patil2022deep,he2025self}. We validate the method in a challenging setting by coupling SI-OCP with robust tube \ac{MPC} for a quadcopter with an all-layer adaptive \ac{DNN} subject to time-varying wind and unmodeled aerodynamics~\cite{lopez2019dynamic, he2025self}.

The main contributions of this paper are as follows: 1)
We develop a conformal uncertainty quantification framework for online dynamics adaptation without requiring state derivative measurements or directly revealed residual labels, through the proposed integral non-conformity score; 2) We extend uncertainty quantification from one-step or instantaneous guarantees to multi-step horizon coverage by using staggered online conformal updates with delayed labels; and 3) We integrate the resulting adaptive robustness margins with safety-critical control, yielding asymptotic long-run safety guarantees, and demonstrate the framework on full all-layer \ac{DNN} dynamics adaptation~\cite{he2025self} within robust tube \ac{MPC}~\cite{lopez2019dynamic}.

\section{Problem Formulation}
\label{sec:prob}


\subsection{System Dynamics}
 Consider a continuous-time nonlinear dynamical system:
\begin{equation}
    \dot{x}(t) = f(x(t), u(t)) + \Delta(t, x(t), u(t)), \label{eq:true_dyn}
\end{equation}
where $x(t) \in \Xcal \subset \mathbb{R}^n$ is the state and $u(t) \in \mathcal{U} \subset \mathbb{R}^m$ is the control input. The function $f: \Xcal \times \Ucal \to \mathbb{R}^n$ denotes the known nominal system dynamics. The term $\Delta: \mathbb{R}_{\ge 0} \times \Xcal \times \Ucal \to \mathbb{R}^n$ denotes the unknown unmodeled dynamics. Both $f$ and $\Delta$ are assumed to be locally Lipschitz continuous in all arguments. We assume that the state $x$ is measured, but the state derivative $\dot{x}$ is unavailable for measurement. We drop the time argument when obvious for brevity.

While the dynamics are continuous in time, the control and uncertainty quantification modules operate at discrete evaluation times $t_k = t_0 + k\Delta t$, $k \in \N$, where $\Delta t$ is the sampling period. We assume control inputs are applied via a zero-order hold, i.e., $u(t) = u(t_k)$ for $t \in [t_k, t_{k+1})$.

\subsection{Adaptation}
Since $\Delta$ is unknown, we approximate it with a known parametric function $F : \mathcal{X} \times \mathcal{U} \times \Theta \to \R^n$, Lipschitz in all arguments. $F$ is parameterized by $\hat\theta \in \Theta \subset \mathbb{R}^d$, which is updated using an adaptive law of the form:
\begin{equation}
    \dot{\hat{\theta}}(t) = \Gamma\big(x(t), u(t), \hat{\theta}(t)\big) ,
    \label{eq:adaptation}
\end{equation}
where $\Gamma: \Xcal \times \Ucal \times \Theta \to \mathbb{R}^d$ is a prescribed adaptation rule. In this paper, we do not constrain $F$ to be of a specific structure. Possible structures include a linear-in-parameter form typically found in adaptive control, $F(x,u,\theta) = \theta^\top\Phi(x,u)$, where $\Phi(x,u)$ is a known vector of basis functions~\cite{parikh2019integral,lavretsky2024robust}, or an all-layer adapted \ac{DNN} where $\theta = \operatorname{vec}(W_1, b_1, \ldots, W_L, b_L)$ and $W_\ell$ and $b_\ell$ are the weights and biases of the $\ell$-th layer~\cite{he2025self}. 

\begin{remark}
   The adaptation law in \eqref{eq:adaptation} covers a broad class of online learning strategies, including gradient-based adaptation~\cite{patil2022deep}, concurrent learning~\cite{chowdhary2010concurrent}, composite learning~\cite{o2022neural}, and recent meta-learned update rules~\cite{he2025self}. Our goal is not to restrict the structure of $\Gamma$, but to quantify the uncertainty of the resulting adaptive predictor.
\end{remark}

We assume that $\Gamma$ is designed so that the parameter estimate $\hat{\theta}(t)$ remains bounded within a compact set for all $t \ge t_0$. Methods for ensuring boundedness include spectral normalization for neural networks~\cite{shi2019neural}, $\sigma-$ or $e-$modifications, or the projection operator used in adaptive control~\cite{lavretsky2024robust}.

To facilitate the analysis of the uncertainty in our approximation, we rewrite the continuous-time dynamics \eqref{eq:true_dyn} as:
\begin{equation}
    \dot{x}(t) = \underbrace{f(x,u) + F(x,u,\hat{\theta})}_{f_{\rm nom}(x,u,\hat{\theta})} + \underbrace{\Delta(t,x,u) - F(x,u,\hat{\theta})}_{d(t,x,u,\hat{\theta})}, \label{eq:lumped_dyn}
\end{equation}
where $f_{\rm nom} : \mathcal{X} \times \mathcal{U} \times \Theta \to \R^n$ denotes the online-computable nominal dynamics and $d : \Rnonneg \times \Xcal \times \Ucal \times \Theta \to \mathbb{R}^n$ denotes the unknown residual disturbance. With a slight abuse of notation, we write $d(t)$ when the dependence on $(x, u, \hat{\theta})$ is clear.
 
\subsection{Online Conformal Prediction}
To quantify the uncertainty of the learned dynamics, we must analyze the accumulation of the residual $d(t)$ over a rolling horizon at each discrete evaluation time $t_k$. 

The sequence of disturbances encountered along a trajectory can be described as a realization of a complex, unknown stochastic process. Let $\mathcal{D}$ denote the unknown probability distribution governing the accumulation of the disturbance over a rolling horizon, conditioned on the state and control inputs. 

Existing uncertainty quantification methods often require assuming $\mathcal{D}$ follows a known parametric form (e.g., Gaussian)~\cite{fan2020deep} or that sequential samples are exchangeable~\cite{angelopoulos2023conformal}. 
Because $d(t)$ depends on unmodeled dynamics and evolving parameter estimates $\hat{\theta}$, the distribution $\mathcal{D}$ is inherently time-varying, non-stationary, and highly correlated with the system trajectories, which prevents the use of standard conformal prediction. As these assumptions are violated in online adaptation settings, we instead use \ac{OCP}, a distribution-free method that provides coverage guarantees without requiring any prior knowledge of $\mathcal{D}$ or assuming exchangeability, which makes it well suited for streaming data generated by adapting systems subject to arbitrary or complex distributions~\cite{angelopoulos2024online}.

All versions of conformal prediction make use of a non-conformity score $S_k$, $k \in \N$, as a metric for the error between a model's prediction and the ground truth. A large score indicates the prediction model is performing poorly. Given a user-specified failure probability $\alpha \in (0,1)$, \ac{OCP} adapts a non-conformity score threshold $q_k$ using subgradient descent on the pinball loss $\ell_\alpha(r) = (1-\alpha)\max\{r,0\} + \alpha\max\{-r,0\}$~\cite{koenker1978regression}. The threshold is updated recursively:
\begin{equation}
    q_{k+1} = q_k + \eta_k (\mathds{1}(S_k > q_k) - \alpha) ,\label{eq:ocp}
\end{equation}
where $\eta_k$ is a positive sequence of step sizes. The purpose of $q_{k+1}$ is to correctly threshold the predicted next value of the score, $S_{k+1}$. The type of guarantee that follows from \ac{OCP} is asymptotic, but is valid for an arbitrary, possibly adversarial sequence of scores bounded by a finite constant $B >0 $.
If the step size $\eta_k$ is non-increasing, the empirical coverage error decreases with iterations $k$:
\begin{equation}
    \left| \frac{1}{K}\sum_{k=1}^K \mathds{1}(S_k \le q_k) - (1-\alpha) \right| \le \frac{B + \eta_1}{\eta_K K}, \label{eq:ocp_bound}
\end{equation}
Taking the limit as $K\to\infty$ yields what we refer to as the \textbf{asymptotic long-run coverage guarantee}:
\begin{equation}
    \lim_{K\to\infty}\frac{1}{K}\sum_{k=1}^K \mathds{1}(S_k \le q_k) = 1-\alpha \label{eq:ocp_guarantee}
\end{equation}

The score function we introduce in \Cref{sec:method} relies on observing the continuous-time variables over a rolling window, formalized below.

\begin{definition}[Rolling History Stack]
    The history stack $\mathcal{H}_{T_p}(t_k)$ contains the continuous trajectories of $x(\tau)$, $u(\tau)$, and $\hat{\theta}(\tau)$ over the time interval $\tau \in [t_k - T_p, t_k]$. 
\end{definition}

\subsection{Safety-Critical Controller}
Given a feedback controller $\pi : \mathcal{X} \times \Theta \to \mathcal{U}$, the closed-loop system is given by \begin{subequations}\label{eq:closed} \eqn{\dot{x} &= f_{\rm nom}(x,\pi(x,\hat{\theta}),\hat{\theta}) + d \\ \dot {\hat \theta} &= \Gamma(x,\pi(x,\hat\theta),\hat\theta)} \end{subequations}
The safety objective is to design a controller for the closed-loop system~\eqref{eq:closed} that renders the safe set $\mathcal{S} \subset \mathcal{X}$ forward invariant, ensuring $x(t) \in \mathcal{S}, ~ \forall t \ge t_0.$
Due to uncertainty in the system dynamics, the controller cannot render the system safe by only using the nominal dynamics~$f_{\rm nom}$. We introduce the notion of a robust predictive safety-critical controller that relies on the uncertainty quantification of the disturbance~$d$.
\begin{definition}[Robust Predictive Safety-Critical Controller]
    A robust predictive safety-critical controller is a locally Lipschitz mapping $\pi_{\rm r} : \mathcal{X} \times \Theta \times \mathbb{R}_{\ge 0} \to \mathcal{U}$ that, at each time $t_k$, selects a control input $u(t_k)$ such that the safe set $\Scal$ is robustly controlled-invariant over the interval $[t_k, t_k + T_p]$ for the continuous-time system~\eqref{eq:closed},
    for any disturbance realization $d(t)$ satisfying $\sup_{t \in [t_k, t_k + T_p]} \|d(t)\| \le \overline{d}_k$.
\end{definition}

\begin{assumption}[Recursive Feasibility] \label{ass:feasibility}
    For the given prediction horizon $T_p$ and the sequence of robustness margins $\{\overline{d}_k\}_{k=1}^\infty$, the set of control inputs $u(t_k)$ that maintain the robust controlled-invariance of $\Scal$ is non-empty for all $k > 0$.
\end{assumption}

\begin{remark} \Cref{ass:feasibility} is a property of the robust controller under the sequence of margins generated by \ac{SIOCP}, rather than an assumption on the conformal predictor itself. For robust MPC, it may be verified by establishing a recursively feasible region parameterized by the disturbance margin using robust terminal set arguments; similar admissible set conditions apply to robust CBF controllers. This verification does not require that the controller operate using a fixed worst-case disturbance bound.
\end{remark}

This class of controllers encompasses single-step safety filters, such as control barrier functions (CBFs) where the prediction horizon is a single step, i.e., $T_p = \Delta t$, as well as multi-step algorithms including robust tube \ac{MPC}~\cite{lopez2019dynamic}, robust model predictive shielding~\cite{li2020robust}, robust \texttt{gatekeeper}~\cite{agrawal2024gatekeeper}, and robust policy CBFs~\cite{knoedler2025safety}. 

Finally, we require the following assumption on the rate of change of the disturbance in order to estimate its upper bound using \ac{OCP}.
\begin{assumption} \label{ass:ddot_bound}
    The residual $d(t)$ is Lipschitz continuous in time with Lipschitz constant $L_d > 0$. That is, $\|\dot{d}(t)\| \le L_d$.
\end{assumption}


\subsection{Problem Statement}
Our overall goal is to compute a sequence of upper bounds $\{\overline d_k\}_{k=1}^\infty$ that correctly bounds the disturbance $d(t)$, which will be used in the robust safety-critical controller $\pi_{\rm r}$ to guarantee safety of the system. 
A trivial solution would be to apply a large constant $\overline d_k = \overline d_{\rm const} > 0$ that loosely bounds $d(t)$. However, the performance of the system would be \textbf{highly conservative}. Instead, we aim to synthesize a dynamic robustness margin without restrictive assumptions on the disturbance distribution, which leads to the following problem statement.


\begin{problem} \label{prob}
Synthesize a dynamic robustness margin $\overline{d}_k$ at each time $t_k$ such that it correctly bounds $d(t)$ with a long-run average rate of at least $1-\alpha$:
\begin{equation}
    \lim_{K \to \infty} \frac{1}{K} \sum_{k=1}^K \Pr_\Dcal \Big( \|d(t)\| \le \overline d_k, ~ \forall t \in [t_k, t_k+T_p] \Big) \ge 1 - \alpha \nonumber
\end{equation}
\end{problem}

With a recursively feasible robust safety-critical controller, solving \Cref{prob} implies the safety of the system will be achieved with a long-run average probability of at least $1-\alpha$.


\section{Safe Dynamics Adaptation with Staggered Integral Online Conformal Prediction}
\label{sec:method}
\begin{figure}[t]
    \centering
    \includegraphics[width=1\linewidth]{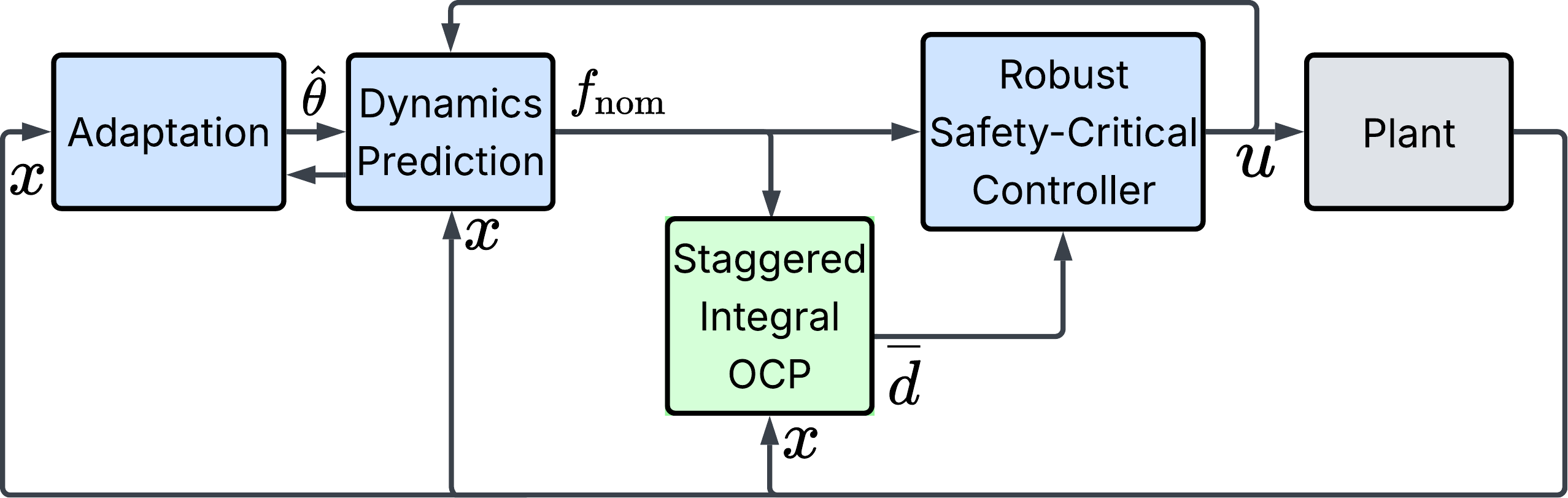}
    \caption{Block diagram of our proposed framework, where we highlight that the \ac{SIOCP} module provides the dynamically adapted estimate of the disturbance upper bound $\overline d_k$ to the robust safety-critical controller.}
    \label{fig:blkdiag}
\end{figure}

In this section, we describe our approach to dynamically quantify the uncertainty of the lumped residual disturbance $d(t)$ to solve~\Cref{prob}. Because the state derivative $\dot{x}$ is unmeasurable, we cannot evaluate the instantaneous disturbance directly. Instead, we quantify the disturbance through its integrated effect on the state trajectory using our method, \textbf{Staggered Integral Online Conformal Prediction (SI-OCP)}, which is a module that quantifies the uncertainty of the dynamics prediction for the robust safety-critical controller, as shown in~\Cref{fig:blkdiag}. 

To bound the uncertainty over the upcoming continuous prediction horizon $t \in [t_k, t_k + T_p]$, the horizon is partitioned into $P = T_p / \Delta t \in \N$ discrete intervals. We utilize a staggered  approach that maintains an array of $P$ staggered non-conformity thresholds which are updated sequentially according to the active thread index $j = k \bmod P$. Inspired by integral concurrent learning~\cite{parikh2019integral}, we introduce a non-conformity score function that captures the residual error using online-computable quantities. Because the true integrated disturbance over the horizon $[t_k, t_k + T_p]$ cannot be evaluated until $t_k + T_p$, we construct the score over the previous horizon $[t_k - T_p, t_k]$. At time $t_k$, the supremum integral score is:
\eqn{ \label{eq:score}
    S_k = \sup_{\tau_1,\tau_2 \in [t_k - T_p, t_k], \tau_1 \le \tau_2}\bigg\| x(\tau_2) - x(\tau_1) \nonumber \\ -\int_{\tau_1}^{\tau_2} f_{\rm nom}\big(x(s),u(s),\hat\theta(s)\big) \mathrm{d}s \bigg\|.
}
We note that the score function can be rewritten as \eqn{S_k =  \sup_{\tau_1,\tau_2 \in [t_k - T_p, t_k], \tau_1 \le \tau_2}\bigg\| \int_{\tau_1}^{\tau_2} d(s) \mathrm{d}s \bigg\|,} where the supremum isolates the integrated disturbance over the sub-interval where it attains its maximal value.

This score serves as the true label for the prediction made by thread $j$ at time $t_k - T_p$. We apply this delayed score in \ac{OCP} to recursively update the active thread's non-conformity threshold $q^{(j)}_k \ge 0$ via subgradient descent on the pinball loss:
\begin{equation}
    q_{k+1}^{(j)} = q_k^{(j)} + \eta_k \Big(\mathds{1}\big(S_k > q_k^{(j)}\big) - \alpha\Big) \label{eq:ocp}
\end{equation}

By decoupling the overlapping prediction horizons into $P$ staggered threads, the staggered approach allows predictions to be made at each step with multi-step coverage. This yields an asymptotic guarantee that the integrated disturbance over the entire upcoming prediction horizon $[t_k, t_k + T_p]$ is bounded by $q_{k+1}^{(j)}$ with a user-defined error rate $\alpha \in (0,1)$~\cite{angelopoulos2024online}. 

We convert the bound on the predicted integrated disturbance~$q^{(j)}_{k+1}$ to a pointwise bound over the full prediction window $[t_k,t_k+T_p]$, depending on the magnitude of the active threshold $q_{k+1}^{(j)}$ relative to the prediction horizon $T_p$ and the Lipschitz constant $L_d$:
\begin{equation}
    \overline{d}_k = \begin{cases} 
    \sqrt{2L_d q_{k+1}^{(j)}} & \text{if } q_{k+1}^{(j)} < \frac{1}{2}L_d T_p^2 \\ 
    \frac{q_{k+1}^{(j)}}{T_p} + \frac{1}{2}L_d T_p & \text{otherwise }
    \end{cases} \label{eq:margin}
\end{equation}

We then provide the dynamic robustness margin $\overline{d}_k$ to the robust safety-critical controller to compute the safe control input $u(t_k)$. The complete execution is summarized in~\Cref{alg}. Before one full horizon $T_p$ has been completed, we use an initial conservative margin $\overline d_{\rm init} \gg 0$. Then, the active thread $j = k \bmod P$ is updated using the score $S_k$.

\begin{remark}
For a sampled-data implementation, let the rolling history stack
$\mathcal{H}_{T_p}(t_k)$ contain data recorded at times
$t_{d,i}=t_{d,0}+i\Delta t_d$, where the data sampling period
$\Delta t_d$ may be smaller than the SI-OCP evaluation period
$\Delta t$. Let $\{t_{d,i}\}_{i=i_0}^{i_0+N}$ denote the samples
contained in the interval $[t_k-T_p,t_k]$, where
$N = T_p/\Delta t_d$. Using the data stored in
$\mathcal{H}_{T_p}(t_k)$, define the cumulative residuals
\begin{equation}
    z_i =
    x(t_{d,i})-x(t_{d,i_0})
    -
    \int_{t_{d,i_0}}^{t_{d,i}}
    f_{\mathrm{nom}}\!\left(x(s),u(s),\hat{\theta}(s)\right)\,ds .
\end{equation}
The sampled approximation of the supremum integral score is then $
    S_k^{\mathrm{samp}}
    =
    \max_{i_0\leq i<j\leq i_0+N}
    \left\|z_j-z_i\right\|.
$
Thus, the score computation reduces to finding the maximum pairwise distance among the cumulative residuals in the history stack.
Constructing the cumulative residuals requires $O(Nn)$ operations,
while evaluation of the pairwise distances is
$O(N^2n)$, where we recall that $n$ is the state dimension. The history stack sampling rate $\Delta t_d$ drives the computational cost independently of the SI-OCP update rate $\Delta t$.
\end{remark}

Now we provide the trajectory coverage guarantee using staggered \ac{OCP} with the integral score function. Unlike the one-step guarantee in~\cite{dixit2023adaptive}, our analysis establishes coverage for the $P$-step target prediction of $S_{k+P}$.

\begin{algorithm}[t]
\small
    \SetKwInOut{Input}{Input}
    \SetKwInOut{Output}{Output}
    \caption{Staggered Integral Online Conformal Prediction (SI-OCP)}
    \label{alg}

    \Input{Target error rate $\alpha$, step size sequence $\{\eta_k\}_{k=1}^\infty$, prediction horizon $T_p$, sampling period $\Delta t$, derivative bound $L_d$, initial margin $\overline{d}_{\rm init} \gg 0$.}
    
    $P \leftarrow T_p / \Delta t$\;
    $q_0^{(i)} \leftarrow q_0 \quad \forall i \in \{0, \dots, P-1\}$\;

    \For{$k$ \text{from} $1$ \text{to} $\infty$}{ 
        Receive state and input trajectories $x(t)$, $u(t)$, and parameter estimates $\hat \theta(t)$, and update history $\mathcal{H}_{T_p}(t_k)$\;
        $j \leftarrow k \bmod P$\;

        \eIf{$t_k \ge T_p$}{ 
            $S_k \leftarrow$ \eqref{eq:score}\;
            $q_{k+1}^{(j)} \leftarrow q_k^{(j)} + \eta_k (\mathds{1}(S_k > q_k^{(j)}) - \alpha)$\;
            $q_{k+1}^{(i)} \leftarrow q_k^{(i)} \quad \forall i \neq j$\;
            $\overline{d}_k \leftarrow $ \eqref{eq:margin}\;
        }{
            $q_{k+1}^{(i)} \leftarrow q_k^{(i)} \quad \forall i \in \{0, \dots, P-1\}$\;
            $\overline{d}_k \leftarrow \overline{d}_{\rm init}$\;
        }

        Compute control input $u(t_k)$ via $\pi_{\rm r}(x(t_k), \hat{\theta}(t_k), \overline{d}_k)$\;
        Apply control input $u(t_k)$\;
    }
\end{algorithm}

\begin{theorem} \label{thm:safety_guarantee}
Let Assumptions \ref{ass:feasibility} and \ref{ass:ddot_bound} hold. Under the robust predictive safety-critical controller $\pi_{\rm r}$, let the robustness margin $\overline{d}_k$ be defined piecewise as in \eqref{eq:margin}. The non-conformity threshold $q^{(j)}_{k+1}$ is updated via staggered \ac{OCP} utilizing $P = T_p / \Delta t$ staggered threads running subgradient descent on the integral score $S_k$, with target error rate $\alpha \in (0,1)$ and step size $\eta_k > 0$. Then, the closed-loop system satisfies the long-run probabilistic safety guarantee:
\begin{equation}
    \lim_{K \to \infty} \frac{1}{K} \sum_{k=1}^K \Pr_\Dcal \Big( x(\tau) \in \mathcal{S}, ~ \forall \tau \in [t_{k}, t_{k} + T_p]\Big) \ge 1 - \alpha \nonumber
\end{equation}
\end{theorem}

\begin{proof}
Let the evaluation times be defined as $t_k = t_0 + k\Delta t$, where the prediction horizon is partitioned into $P = T_p / \Delta t$ discrete intervals. The \ac{SIOCP} algorithm maintains an array of $P$ staggered non-conformity thresholds, updated according to the active thread index $j = k \bmod P$. At time $t_k$, the algorithm evaluates the integral score $S_k$ over the interval $[t_k - T_p, t_k]$. This score serves as the true label for the prediction that was made by thread index $j$ at time $t_k - T_p$.

Let $d^*_{k+P}$ be the true maximum disturbance magnitude realized over the target interval $[t_k, t_{k+P}]$, attained at time $t^* \in [t_k, t_{k+P}]$. Define the unit vector $v = d(t^*) / d^*_{k+P}$ and the scalar projection $g(t) = v^T d(t)$. By \Cref{ass:ddot_bound}, the derivative of the projection is bounded as $|\dot{g}(t)| \le L_d$. Integrating from $t^*$ yields $g(t) \ge d^*_{k+P} - L_d|t - t^*|$. 

We now consider two cases for deriving the disturbance bound based on the \ac{OCP} quantile $q^{(j)}_{k+1}$. The two cases distinguish whether the Lipschitz bound allows the disturbance to decrease to zero within the prediction horizon $T_p$.

\textbf{Case 1:} $T_p \ge d^*_{k+P}/L_d$. Integrating the projection over a contiguous subset of length $d^*_{k+P}/L_d$ yields
$
    S_{k+P} \ge \frac{(d^*_{k+P})^2}{2L_d}.
$
If the target score is bounded by the active threshold ($S_{k+P} \le q^{(j)}_{k+1}$), isolating $d^*_{k+P}$ gives $d^*_{k+P} \le \sqrt{2L_d q^{(j)}_{k+1}}$.

\textbf{Case 2:} $T_p < d^*_{k+P}/L_d$. Integrating over the full prediction horizon yields a truncated trapezoid
$
    S_{k+P} \ge d^*_{k+P} T_p - \frac{1}{2}L_d T_p^2.
$
If $S_{k+P} \le q^{(j)}_{k+1}$, isolating $d^*_{k+P}$ gives $d^*_{k+P} \le \frac{q^{(j)}_{k+1}}{T_p} + \frac{1}{2}L_d T_p$.

Applying the piecewise definition of $\overline{d}_k$ derived from $q^{(j)}_{k+1}$ guarantees that if the conformal constraint holds ($S_{k+P} \le q^{(j)}_{k+1}$), the pointwise disturbance over the entire upcoming horizon $[t_k, t_k + T_p]$ is bounded by $\overline{d}_k$.

Because the $P$ staggered threads are updated independently, the standard empirical coverage guarantee \cite[Theorem 1]{angelopoulos2024online} applies to each thread sequence, where $\Kcal_j(K) = \{k \in \{1,\dots,K\}~|~ k \bmod P = j\}$ denotes the set of update indices for thread $j$ and $K_j(K) = |\Kcal_j(K)|$:
\begin{equation}
    \lim_{K \to \infty} \frac{1}{K_j(K)} \sum_{k\in\Kcal_j(K)} \mathds{1}(S_{k+P} \le q^{(j)}_{k+1}) = 1 - \alpha.
\end{equation}

Summing the convergent limits across all $P$ staggered sequences yields the global asymptotic guarantee:
\eqn{
    &\lim_{K \to \infty} \frac{1}{K} \sum_{k=1}^K \mathds{1}(S_{k+P} \le q^{(k \bmod P)}_{k+1}) \nonumber \\ &= \lim_{K\to\infty}\sum_{j=0}^{P-1} \frac{K_j(K)}{K} \left[ \frac{1}{K_j(K)} \sum_{k\in\Kcal_j(K)} \mathds{1}(S_{k+P} \le q^{(j)}_{k+1})\right] \nonumber \\ &= \sum_{j=0}^{P-1} \frac{1}{P}(1-\alpha) = 1 - \alpha.
}

Taking the expected value of both sides and applying the relation $\Pr(\cdot) = \mathbb{E}[\mathds{1}(\cdot)]$ yields the probabilistic limit:
\eqn{
    \lim_{K \to \infty} \frac{1}{K} \sum_{k=1}^K \Pr_\Dcal(S_{k+P} \le q^{(k \bmod P)}_{k+1}) = 1 - \alpha. \label{eq:prob_proof}
}
Under the robust predictive controller $\pi_{\rm r}$, bounding the full prediction horizon ensures the planned trajectory satisfies $x(\tau) \in \mathcal{S}$ for all $\tau \in [t_k, t_k + T_p]$.  \Cref{ass:feasibility} establishes
\begin{equation} \label{eq:safety_implication}
    \mathds{1}(S_{k+P} \le q^{(j)}_{k+1}) \le \mathds{1}\Big(x(\tau) \in \mathcal{S}, ~ \forall \tau \in [t_k, t_k + T_p]\Big).
\end{equation}
Combining~\eqref{eq:safety_implication} with~\eqref{eq:prob_proof} completes the proof.
\end{proof}
\begin{remark}
    The asymptotic safety guarantee extends to arbitrary positive step size sequences $\{\eta_k\}$. Following~\cite[Theorem 2]{angelopoulos2024online}, the empirical OCP coverage error is:
    \begin{equation}
        \left| \frac{1}{K} \sum_{k=1}^K \mathds{1}(S_k \le q_k) - (1 - \alpha) \right| \le \frac{B + \overline \eta_k}{K} \|\Delta_{1:K}\|_1,
    \end{equation}
    where $\Delta_1 = \eta_1^{-1}$ and $\Delta_k = \eta_k^{-1} - \eta_{k-1}^{-1}$ for $k \ge 2$, and $\overline{\eta}_K = \max_{1 \le k \le K} \eta_k$ is the maximum step size over the sequence. If $\eta_k$ is adaptively increased $N_K$ times during execution to react to sudden distribution shifts, the sequence variation is bounded by $\|\Delta_{1:K}\|_1 \le \frac{2N_K}{\min_{1 \le k \le K} \eta_k}$. Taking the expectation as in Theorem \ref{thm:safety_guarantee}, the long-run probability of safety still converges to $1-\alpha$ if the number of step size resets grows sublinearly.
\end{remark}

\section{Simulation}
\label{sec:sim}
The proposed algorithm is evaluated with a quadcopter navigating to a target region while avoiding obstacles under time-varying, complex aerodynamic disturbances. The adaptation module comprises an all-layer adapted \ac{DNN} updated via Self-Supervised Meta Learning (SSML)~\cite{he2025self} and represents a complex nonlinear-in-parameter prediction model. The controller is the Dynamic Tube Model Predictive Control (DT-MPC) algorithm, which incorporates the disturbance margin $\overline d_k$ into the robust tube synthesis~\cite{lopez2019dynamic}. We reiterate that our method applies to any adaptation law of the form~\eqref{eq:adaptation}, including those used with simple linear-in-parameter models commonly found in adaptive and safety-critical control.

\subsubsection{Simulation Setup and System Dynamics}

We consider a 3D quadcopter with unmodeled aerodynamic forces. The system state $x \in \mathbb{R}^8$ and control input $u \in \mathbb{R}^3$ are defined as:
\begin{equation}
    x = [r^\top, v^\top, \phi, \vartheta]^\top, \qquad u = [p, q, T]^\top,
\end{equation}
where $r \in \mathbb{R}^3$ is position, $v \in \mathbb{R}^3$ is velocity, $\phi$ and $\vartheta$ represent roll and pitch angles, respectively, $p$ and $q$ are body-frame angular rates, and $T$ is the thrust. The dynamics are:
\begin{subequations}
\begin{align}
    \dot{r} &= v, \quad \dot{\phi} = p, \quad \dot{\vartheta} = q, \label{eq:pos_dyn} \\
    m\dot{v} &= mg + \begin{bmatrix} \sin\vartheta \\ -\cos\vartheta\sin\phi \\ \cos\vartheta\cos\phi \end{bmatrix} T + \Delta_v(t,x), \label{eq:vel_dyn} \end{align}
\end{subequations}
where $\Delta_v \in \mathbb{R}^3$ is the unmodeled force disturbance. The unmodeled dynamics are $\Delta(t, x) = [0_{1\times 3}, \Delta_v^\top/m, 0_{1\times 2}]^\top$, with mass $m=1.0$ kg. $g$ is gravitational acceleration.

\subsubsection{Unmodeled Dynamics}

The simulated unmodeled dynamics include aerodynamic drag as a function of the relative airspeed in the body frame, coupling the vehicle attitude, velocity, and the ambient wind field. Let $v_{\rm wind}(r, t) \in \mathbb{R}^3$ represent a time-varying, spatially dependent wind field:
\begin{equation}
    v_{\rm wind}(r,t) = \begin{bmatrix} 2\sin(0.5t) + \sin(2t) + 0.5r_x \\ 2.4\cos(0.4t) + 1.2\cos(1.8t) + 0.5r_y \\ \sin(0.3t) + 0.2r_z \end{bmatrix}.
\end{equation}
The relative airspeed is defined as $v_{\rm rel} = v - v_{\rm wind}(r, t)$. Let $R(\phi, \vartheta) \in SO(3)$ be the rotation matrix from the body frame to the world frame, yielding the body-frame relative velocity $v_b = R(\phi, \vartheta)^\top v_{\rm rel}$. Aerodynamic drag is parameterized by a diagonal matrix $D_{\rm body} = \operatorname{diag}(0.3, 0.3, 0.6)$. The resulting aerodynamic force in the world frame is:
\begin{equation}
    \Delta_v(t, x) = -m R(\phi, \vartheta) D_{\rm body} (v_b\|v_b\|) + \epsilon,
\end{equation}
where $\epsilon \sim \mathcal{N}(0, \Sigma)$ denotes additive Gaussian noise with standard deviations $\sigma_{x} = \sigma_y = 0.2$ and $\sigma_z = 0.1$. This model is used in the simulation code but is not known by the controller.

\subsubsection{Neural Network Model and Adaptation}
\begin{figure}[t]
    \centering
    \includegraphics[width=\linewidth]{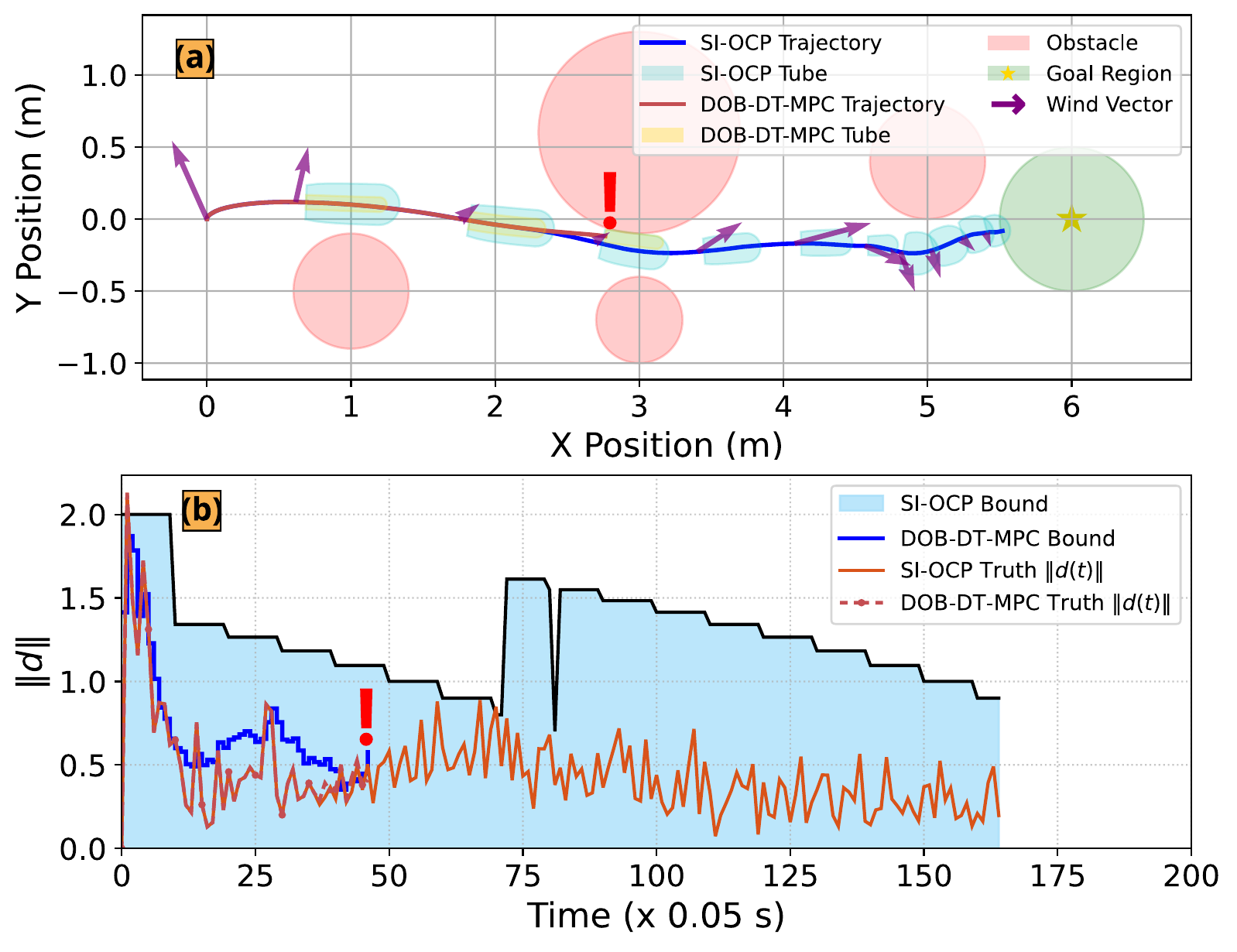}
    \caption{(a) Top-down view of the trajectory of the quadcopter flying around obstacles subject to wind disturbances and unmodeled aerodynamics. The tubes from DT-MPC are shown at varying intervals and are narrow enough to pass through the obstacles because \ac{SIOCP} estimates the error tightly. The  disturbance-observer-based (DOB)-DT-MPC trajectory has a lower disturbance bound but exhibits higher miscoverage and results in a collision with the obstacle, shown by the red exclamation point. (b) The \ac{SIOCP} predicted error bound correctly bounds the true disturbance 98.77\% of the time, which is above the desired coverage rate of 90\%, while the DOB-DT-MPC coverage rate is only 50\%. }
    \label{fig:main_fig}
    \vspace{-10pt}
\end{figure}
\begin{figure}[t]
    \centering
    \includegraphics[width=\linewidth]{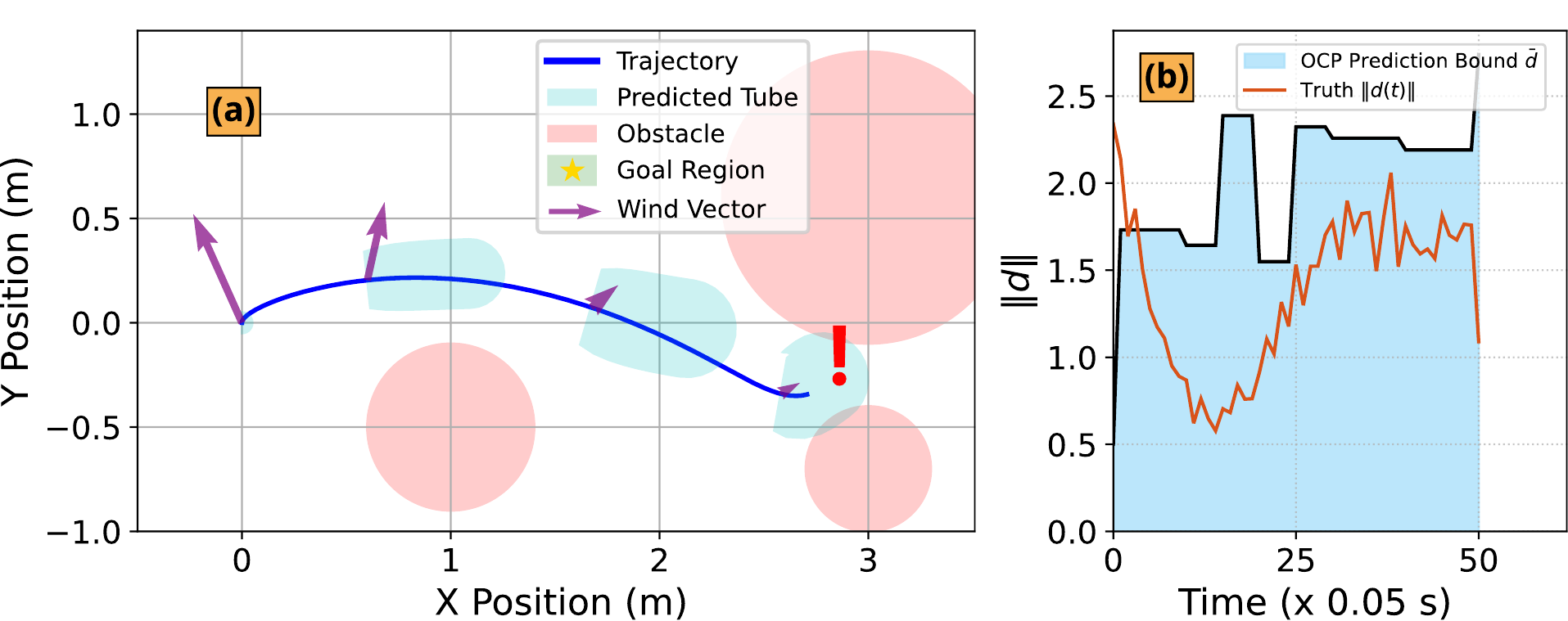}
    \caption{(a) When adaptation is disabled, the true disturbance is larger as the offline-trained weights do not provide a good approximation of the unmodeled dynamics. As such, in (b), the estimated disturbance bound from \ac{SIOCP} covers the true disturbance 94.0\% of the time, which increases the tube size and prevents the quadcopter from flying through the narrow gap, denoted by the red exclamation point. Note that the visual effect of a tube intersecting with a spherical obstacle is due to the fact that the figure is a top-down view of a 3D environment. The tube does not intersect with the sphere in 3D. }
    \label{fig:no_adapt}
    \vspace{-10pt}
\end{figure}
The unmodeled dynamics are predicted online using a 4-layer multi-layer perceptron with ReLU activation functions. The input vector is defined as $\xi = [v_x, v_y, v_z, \phi, \vartheta]^\top \in \mathbb{R}^5$. Hidden layer propagation is defined as:
\begin{align}
    h^{(1)} &= \text{ReLU}(W_1 \xi + b_1), \\
    h^{(l+1)} &= \text{ReLU}(W_{l+1} h^{(l)} + b_{l+1}), \quad \forall l \in \{1,2\},
\end{align}
resulting in the parameterized acceleration $F_{\text{nn}}(\xi, \hat{\theta}) = W_4 h^{(3)} + b_4 \in \mathbb{R}^3$. The layer architecture follows a $5 \to 50 \to 50 \to 50 \to 3$ progression. The parameters $\hat{\theta} \in \mathbb{R}^d$ are updated online to minimize the finite-difference acceleration residual $\varepsilon_{\text{acc}} = (v(t) - v(t-\Delta t))/\Delta t - \dot v_{\text{nom}}$ via the continuous update law
$
    \dot{\hat{\theta}} = \gamma\,J^\top \varepsilon_{\text{acc}} - \lambda(\hat{\theta} - \theta_0),
$
where $\gamma = 5.0$ is the learning gain, $J = \frac{\partial F_{\text{nn}}}{\partial \hat{\theta}}(\xi;\hat{\theta})$ is the Jacobian, and $\lambda = 0.1$ is a regularization parameter relative to the meta-learned prior $\theta_0$. Network parameters are constrained via spectral normalization to $\|\hat{\theta}\|_2 \le 10$. 
Compared to~\cite{he2025self}, we omit the trajectory-tracking error term that also drives parameter adaptation because our simulation does not use a reference trajectory.

\subsubsection{Robust Safety-Critical Control}

The \ac{SIOCP} algorithm synthesizes the residual bound $\overline d_k$ for the \ac{DNN} approximation error and time-varying unmodeled disturbances. With a specified failure rate $\alpha = 0.1$, the resulting margin is utilized by a 10-step-horizon DT-MPC operating at $20$ Hz, which generates robust tube trajectories with time-varying tube radii $\Phi(t) \in \Rnonneg$~\cite{lopez2019dynamic}. 
The scenario requires the quadcopter to transition from $r_{\rm start} = [-2.0, 0.0, 1.0]^\top$ m to $r_{\rm goal} = [7.0, 0.0, 1.0]^\top$ m within a $5.0$ s duration. The flight path is obstructed by three spherical obstacles. A pair of obstacles with radii $0.7$ m and $0.3$ m forms a constrained corridor at $r_x = 3.0$ m, limiting feasible tube expansion. A third obstacle with a $0.4$ m radius is at $r_x = 5.5$ m. 

The integral \ac{OCP} algorithm quantifies parameter adaptation errors and stochastic disturbances in real time. By utilizing the adaptive margin $\overline{d}_k$, the DT-MPC modulates the tube sizing $\Phi$ to maintain safety within the constrained environment while providing robust guarantees against stochastic disturbances and DNN approximation errors.

\subsubsection{Results}

We consider three scenarios: 1) a nominal case with \ac{DNN} adaptation enabled, 2) an off-nominal case in which adaptation is disabled and the \ac{DNN} is frozen at its offline meta-learned weights, and 3) a comparison with a disturbance observer from~\cite{dacs2022robust} providing the disturbance bound to DT-MPC, which we denote disturbance-observer-based DT-MPC (DOB-DT-MPC). Additional implementation details, code, and animations are available in our public repository.\footnote{ \href{https://github.com/dcherenson/staggered-integral-ocp}{https://github.com/dcherenson/staggered-integral-ocp}}

As shown in \Cref{fig:main_fig}, with adaptation turned on, the \ac{DNN} reduces the approximation error of the unmodeled dynamics and the resulting estimated disturbance bound is lower, allowing the quadcopter to fit through the narrow gap between obstacles with a smaller tube. The DOB-DT-MPC case fails to reach the goal and collides with an obstacle, as its estimated disturbance bound did not accurately bound the true disturbance near the obstacle. The average runtime for the \ac{SIOCP} iterations is 18 ms.

In the non-adaptive case in \Cref{fig:no_adapt}, the lumped disturbance is high because the frozen \ac{DNN} does not accurately model the unmodeled dynamics $\Delta$. As such, the estimated disturbance bound captures this and the larger tube size prevents the quadcopter from going through the narrow gap.


\section{Conclusion}

This paper proposes Staggered Integral Online Conformal Prediction (SI-OCP), a novel algorithm for quantifying the lumped uncertainty of unmodeled disturbances and adaptive learning errors in the dynamics model. By utilizing a non-conformity score function based on an integral over a rolling window of recorded data, our proposed approach avoids the need for state derivative measurements and provides long-run probabilistic safety guarantees without requiring exchangeability or prior knowledge of the disturbance distribution. The effectiveness of this approach is validated through simulation of an all-layer \ac{DNN} adaptive quadcopter navigating a constrained environment under unmodeled aerodynamics and time- and spatially-varying wind disturbances, demonstrating that the algorithm successfully modulates robust tubes using the \ac{SIOCP}-derived disturbance bound to maintain safety, including a case in which adaptation is disabled and the \ac{DNN} poorly models the dynamics.

Future work will attempt to relax the assumption on the disturbance derivative and allow partial observability of the states, which would necessitate quantifying both the disturbance error and  state estimation error. Additionally, methods for actively reducing the uncertainty in the model can be implemented as a form of dual control~\cite{naveed2026formal}.

\bibliographystyle{IEEEtran}
\bibliography{biblio, IEEEabrv}
\end{document}